\documentclass[a4paper]{llncs}

\usepackage{amsmath}
\usepackage{amssymb}
\usepackage{amsfonts}
\usepackage{pifont}
\usepackage{listings}
\usepackage{graphicx}
\usepackage{xspace}
\usepackage{pgf}
\usepackage[noend]{algpseudocode}
\usepackage{algorithm}
\usepackage{multicol}
\usepackage{appendix}
\usepackage{caption,subcaption}
\DeclareCaptionType{copyrightbox}
\usepackage{subfig}
\usepackage{multirow}
\usepackage{framed}
\usepackage{tikz}
\usetikzlibrary{arrows, automata, shapes, calc}

\newcommand{\xmark}{\ding{55}}
\newcommand{\tick}{\checkmark}
\newcommand{\todo}[1]{{\bf TODO:} #1}

\algrenewcommand\algorithmicrequire{\textbf{Precondition:}}
\algrenewcommand\algorithmicensure{\textbf{Postcondition:}}

\makeatletter
\pgfdeclareshape{datastore}{
  \inheritsavedanchors[from=rectangle]
  \inheritanchorborder[from=rectangle]
  \inheritanchor[from=rectangle]{center}
  \inheritanchor[from=rectangle]{base}
  \inheritanchor[from=rectangle]{north}
  \inheritanchor[from=rectangle]{north east}
  \inheritanchor[from=rectangle]{east}
  \inheritanchor[from=rectangle]{south east}
  \inheritanchor[from=rectangle]{south}
  \inheritanchor[from=rectangle]{south west}
  \inheritanchor[from=rectangle]{west}
  \inheritanchor[from=rectangle]{north west}
  \backgroundpath{
    \southwest \pgf@xa=\pgf@x \pgf@ya=\pgf@y
    \northeast \pgf@xb=\pgf@x \pgf@yb=\pgf@y
    \pgfpathmoveto{\pgfpoint{\pgf@xa}{\pgf@ya}}
    \pgfpathlineto{\pgfpoint{\pgf@xb}{\pgf@ya}}
    \pgfpathmoveto{\pgfpoint{\pgf@xa}{\pgf@yb}}
    \pgfpathlineto{\pgfpoint{\pgf@xb}{\pgf@yb}}
 }
}
\makeatother

\begin{document}

\setlength{\pdfpageheight}{\paperheight}
\setlength{\pdfpagewidth}{\paperwidth}





\title{Unrestricted Termination and Non-Termination Arguments for Bit-Vector Programs}

\author{Cristina David\and Daniel Kroening\and Matt Lewis}
\institute{University of Oxford}



\maketitle

\begin{abstract}
Proving program termination is typically done by finding a well-founded
ranking function for the program states.  Existing termination
provers typically find ranking functions using either linear algebra or
templates.  As such they are often restricted to finding linear ranking
functions over mathematical integers.  This class of functions is
insufficient for proving termination of many terminating programs, and
furthermore a termination argument for a program operating on mathematical
integers does not always lead to a termination argument for the same program
operating on fixed-width machine integers.
We propose a termination analysis able 
to generate nonlinear, lexicographic ranking functions and
nonlinear recurrence sets that are correct for fixed-width machine arithmetic
and floating-point arithmetic
Our technique is based on a reduction from program \emph{termination} to 
second-order \emph{satisfaction}. We provide 
formulations for termination and non-termination in a fragment of second-order logic 
with restricted quantification which is decidable over finite domains \cite{kalashnikov}.
The resulted technique is a sound and complete analysis for the termination 
of finite-state programs with fixed-width
integers and IEEE floating-point arithmetic.
\end{abstract}



\section{Introduction}\label{sec:intro}

The halting problem has been of central interest to computer scientists
since it was first considered by Turing in 1936~\cite{turing}.  Informally,
the halting problem is concerned with answering the question
``does this program run forever, or will it eventually terminate?''

Proving program termination is typically done by finding a \emph{ranking
function} for the program states, i.e.~a monotone map from the program's
state space to a well-ordered set.  Historically, the search for ranking
functions has been constrained in various syntactic ways, leading to
incompleteness, and is performed over abstractions that do not soundly
capture the behaviour of physical computers.  In this paper, we present a
sound and complete method for deciding whether a program with a fixed amount of
storage terminates.  Since such programs are necessarily finite state, our
problem is much easier than Turing's, but is a better fit for analysing computer
programs.

When surveying the area of program termination chronologically, we observe
an initial focus on monolithic approaches based on a single measure shown to
decrease over all program
paths~\cite{DBLP:conf/vmcai/P04,DBLP:conf/cav/BradleyMS05}, followed by more
recent techniques that use termination arguments based on Ramsey's
theorem~\cite{DBLP:conf/lpe/CodishG03,DBLP:conf/lics/PodelskiR04,DBLP:conf/pldi/CookPR06}.
The latter proof style builds an argument that a transition relation is disjunctively well founded
by composing several small well-foundedness arguments.
The main benefit of this approach is
the simplicity of local termination measures in contrast to global ones. 
For instance, there are cases in which linear arithmetic suffices when using
local measures, while corresponding global measures require nonlinear
functions or lexicographic orders.

One drawback of the Ramsey-based approach is that the validity of the
termination argument relies on checking the \emph{transitive closure} of the
program, rather than a single step.  As such, there is experimental evidence
that most of the effort is spent in reachability
analysis~\cite{DBLP:conf/pldi/CookPR06,DBLP:conf/cav/KroeningSTW10},
requiring the support of powerful safety checkers: there is a trade-off
between the complexity of the termination arguments and that of checking
their validity.

As Ramsey-based approaches are limited by the state of the art in safety
checking, recent research shifts back to more complex termination arguments
that are easier to
check~\cite{DBLP:conf/cav/KroeningSTW10,DBLP:conf/tacas/CookSZ13}. 
Following the same trend, we investigate its extreme: \emph{unrestricted}
termination arguments.  This means that our ranking functions may involve
nonlinearity and lexicographic orders: we do not commit to any particular
syntactic form, and do not use templates.  Furthermore, our approach allows
us to \emph{simultaneously} search for proofs of \emph{non-termination},
which take the form of recurrence sets.

Figure~\ref{fig:handletable} summarises the related work with respect to the
restrictions they impose on the transition relations as well as the form of
the ranking functions computed.  While it supports the observation that the
majority of existing termination analyses are designed for linear programs
and linear ranking functions, it also highlights another simplifying
assumption made by most state-of-the-art termination provers: that
bit-vector semantics and integer semantics give rise to the same termination
behaviour.  Thus, most existing techniques treat fixed-width machine integers
(bit-vectors) and IEEE floats as mathematical integers and reals,
respectively~\cite{DBLP:conf/pldi/CookPR06,DBLP:conf/popl/Ben-AmramG13,DBLP:conf/vmcai/P04,DBLP:conf/atva/HeizmannHLP13,DBLP:conf/vmcai/BradleyMS05,DBLP:conf/cav/KroeningSTW10}.

By assuming bit-vector semantics to be identical to integer semantics, these
techniques ignore the wrap-around behaviour caused by overflows, which can
be unsound.  In Section~\ref{sec:motivation}, we show that integers and
bit-vectors exhibit incomparable behaviours with respect to termination,
i.e.~programs that terminate for integers need \emph{not} terminate for
bit-vectors and vice versa.  Thus, abstracting bit-vectors with integers may
give rise to {\em unsound} and {\em incomplete} analyses.


We present a technique that treats linear and nonlinear programs uniformly and
it is not restricted to finding linear ranking functions, but can
also compute lexicographic nonlinear ones. 
Our approach is constraint-based and relies on second-order formulations 
of termination and non-termination.
The obvious issue is that,
due to its expressiveness, second-order logic is very difficult to reason
in, with many second-order theories becoming undecidable even when the
corresponding first-order theory is decidable.
To make solving our constraints tractable, we formulate termination and non-termination inside a fragment of second-order logic with restricted quantification, 
for which we have built a solver in \cite{kalashnikov}.
Our method is sound and
complete for bit-vector programs -- for any program, we find
a proof of either its termination or non-termination.


%

\begin{figure*}
\centering
 \begin{tabular}{|ll||c|c|c|c|c|c|c|c|}
 \hline
  & & \multicolumn{8}{c|}{Program} \\
  & & \multicolumn{2}{c|}{Rationals/Integers} & \multicolumn{2}{c|}{Reals} & \multicolumn{2}{c|}{Bit-vectors} & \multicolumn{2}{c|}{Floats} \\
  & & L & NL & L & NL & L & NL & L & NL \\
  \hline
  \hline
  \multirow{4}{*}{Ranking} & Linear lexicographic &  \cite{DBLP:conf/popl/Ben-AmramG13,DBLP:conf/cav/BradleyMS05,DBLP:conf/tacas/CookSZ13,DBLP:conf/vmcai/P04} & - & \cite{DBLP:conf/tacas/LeikeH14} & - &\checkmark&\checkmark&\checkmark&\checkmark\\
   & Linear non-lexicographic & \cite{DBLP:conf/pldi/CookPR06,DBLP:conf/cav/LeeWY12,DBLP:conf/atva/HeizmannHLP13,DBLP:conf/vmcai/BradleyMS05,DBLP:conf/cav/KroeningSTW10} & \cite{DBLP:conf/vmcai/BradleyMS05} & \cite{DBLP:conf/tacas/LeikeH14} & - & \checkmark~ \cite{DBLP:conf/tacas/CookKRW10} &\checkmark~ \cite{DBLP:conf/tacas/CookKRW10}&\checkmark&\checkmark\\
   & Nonlinear lexicographic & - & - & - & - &\checkmark&\checkmark&\checkmark&\checkmark\\
   & Nonlinear non-lexicographic & \cite{DBLP:conf/vmcai/BradleyMS05} &  \cite{DBLP:conf/vmcai/BradleyMS05} & - & - &\checkmark&\checkmark&\checkmark&\checkmark\\
   \hline
 \end{tabular}

 \caption{Summary of related termination analyses. Legend: \checkmark = we can handle; - = no available works; L = linear; NL = nonlinear.} \label{fig:handletable}
\end{figure*}

The main contributions of our work can be summarised as follows:
\begin{itemize}

\item  We rephrased the termination and non-termination problems as
second-order 
satisfaction problems.  This formulation
captures the (non-)termination
properties of all of the loops in the program, including
nested loops.  We can use this to analyse all the loops at once,
or one at a time.  Our treatment handles termination and non-termination uniformly:
both properties are captured in the same second-order formula.	

\item We designed a bit-level accurate technique for computing ranking
functions and recurrence sets that correctly accounts for the wrap-around
behaviour caused by under- and overflows in bit-vector and floating-point arithmetic.  Our
technique is not restricted to finding linear ranking functions, but can
also compute lexicographic nonlinear ones.



\item We implemented our technique and tried it on a selection of programs
handling both bit-vectors and floats. In our implementation we made use of 
a solver for a fragment of second-order logic with restricted quantification that is decidable 
over finite domains \cite{kalashnikov}. 

\end{itemize} 

{\bf Limitations.} Our algorithm proves termination for transition systems
with finite state spaces.  The (non-)termination proofs take the form of
ranking functions and program invariants that are expressed in a
quantifier-free language.  This formalism is powerful enough to handle a
large fragment of C, but is not rich enough to analyse code that uses
unbounded arrays or the heap.  Similar to other termination
analyses~\cite{DBLP:conf/tacas/CookSZ13}, we could attempt to alleviate the
latter limitation by abstracting programs with heap to arithmetic
ones~\cite{DBLP:conf/popl/MagillTLT10}.  Also, we have not yet added support
for recursion or \texttt{goto} to our encoding.

\section{Motivating Examples} \label{sec:motivation}

Figure~\ref{fig:handletable} illustrates the most common simplifying
assumptions made by existing termination analyses:
\begin{itemize}
\item[(i)] programs use only linear arithmetic.
\item[(ii)] terminating programs have termination arguments expressible in linear arithmetic.
\item[(iii)] the semantics of bit-vectors and mathematical integers are equivalent.
\item[(iv)] the semantics of IEEE floating-point numbers and mathematical reals are equivalent.
\end{itemize}  

To show how these assumptions are violated by even simple programs, we draw
the reader's attention to the programs in Figure~\ref{fig:motivation} and
their curious properties:

\begin{itemize}

\item Program (a) breaks assumption (i) as it makes use of the bit-wise $\&$ operator.
Our technique finds that an admissible ranking function is the linear
function $R(x) = x$, whose value decreases with every iteration, but cannot
decrease indefinitely as it is bounded from below.  This example also
illustrates the lack of a direct correlation between the linearity of a
program and that of its termination arguments.

\item Program (b) breaks assumption (ii), in that it has no linear ranking
function.  We prove that this loop terminates by finding the nonlinear
ranking function $R(x) = |x|$.

\item Program (c) breaks assumption (iii).  This loop is terminating for
bit-vectors since $x$ will eventually overflow and become negative. 
Conversely, the same program is non-terminating using integer arithmetic
since $x > 0 \rightarrow x+1 > 0$ for any integer $x$.

\item Program (d) also breaks assumption (iii), but ``the other way'': it
terminates for integers but not for bit-vectors.  If each of the variables
is stored in an unsigned $k$-bit word, the following entry state will lead
to an infinite loop:
$$ M = 2^k - 1,\quad N = 2^k - 1,\quad i = M,\quad j = N-1 $$

\item Program (e) breaks assumption (iv): it terminates for reals but not
for floats.  If $x$ is sufficiently large, rounding error will cause the
subtraction to have no effect.

\item Program (f) breaks assumption (iv) ``the other way'': it terminates
for floats but not for reals.  Eventually $x$ will become sufficiently small
that the nearest representable number is $0.0$, at which point it will be
rounded to $0.0$ and the loop will terminate.

\end{itemize}

Up until this point, we considered examples that are not soundly treated by
existing techniques as they don't fit in the range of programs addressed by
these techniques.  Next, we look at some programs that are handled by
existing termination tools via dedicated analyses.  We show that our method
handles them uniformly, without the need for any special treatment.
\begin{itemize}

\item Program (g) is a linear program that is shown
in~\cite{DBLP:conf/tacas/CookSZ13} not to admit (without prior manipulation)
a lexicographic linear ranking function.  With our technique we can find the
nonlinear ranking function $R(x) = |x|$.


\item Program (h) illustrates conditional termination.  When proving program
termination we are simultaneously solving two problems: the search for a
termination argument, and the search for a supporting
invariant~\cite{DBLP:conf/cav/BrockschmidtCF13}.  For this loop, we find the
ranking function $R(x) = x$ together with the supporting invariant $y=1$.

\item In the terminology of \cite{DBLP:conf/tacas/LeikeH14}, program (i)
admits a \emph{multiphase} ranking function, computed from a multiphase
ranking template.  Multiphase ranking templates are targeted at programs
that go through a finite number of phases in their execution.  Each phase is
ranked with an affine-linear function and the phase is considered to be
completed once this function becomes non-positive.

In our setting this type of programs does not need special treatment, as we
can find a nonlinear lexicographic ranking function $R(x, y, z) = (x < y,
z)$.\footnote{This termination argument is somewhat subtle.  The Boolean
values $\mathit{false}$ and $\mathit{true}$ are interpreted as 0 and 1,
respectively.  The Boolean $x < y$ thus eventually decreases, that is to say
once a state with $x \geq y$ is reached, $x$ never again becomes greater
than $y$.  This means that as soon as the ``else'' branch of the if
statement is taken, it will continue to be taken in each subsequent
iteration of the loop.  Meanwhile, if $x < y$ has not decreased (i.e., we
have stayed in the same branch of the ``if''), then $z$ does decrease. 
Since a Boolean only has two possible values, it cannot decrease
indefinitely.  Since $z > 0$ is a conjunct of the loop guard, $z$ cannot
decrease indefinitely, and so $R$ proves that the loop is well founded.}


\end{itemize}
As with all of the termination proofs presented in this paper, the ranking
functions above were all found completely automatically.

\begin{figure*}[h]
\hspace*{-2.5cm}
\centering
\begin{tabular}{ccc}
\begin{subfigure}[b]{0.45\textwidth}
\begin{lstlisting}
while (x > 0) {
  x = (x - 1) & x;
}
\end{lstlisting}
\caption{Taken from~\cite{DBLP:conf/tacas/CookKRW10}.}
 \label{fig:motivation.a}
\end{subfigure}%

&

\begin{subfigure}[b]{0.45\textwidth}
\begin{lstlisting}
while (x != 0) {
  x = -x / 2;
}
\end{lstlisting}
\caption{}
 \label{fig:motivation.b}
\end{subfigure}%

&

\begin{subfigure}[b]{0.45\textwidth}
\begin{lstlisting}[language=C]
while(x > 0) {
  x++;
}
 \end{lstlisting}
\caption{}
 \label{fig:motivation.c}
\end{subfigure} \\

\hline

\begin{subfigure}[b]{0.45\textwidth}
\begin{lstlisting}
while (i<M || j<N) {
  i = i + 1;
  j = j + 1;
}
\end{lstlisting}
\caption{Taken from~\cite{DBLP:conf/sigsoft/Nori013}}
 \label{fig:motivation.d}
\end{subfigure} 

&

\begin{subfigure}[b]{0.45\textwidth}
\begin{lstlisting}
float x;

while (x > 0.0) {
  x -= 1.0;
}
\end{lstlisting}
\caption{}
 \label{fig:motivation.e}
\end{subfigure} 

&

\begin{subfigure}[b]{0.45\textwidth}
\begin{lstlisting}
float x;

while (x > 0.0) {
  x *= 0.5;
}
\end{lstlisting}
\caption{}
 \label{fig:motivation.f}
\end{subfigure} \\
\hline

\begin{subfigure}[b]{0.45\textwidth}
\begin{lstlisting}
while (x != 0) {
  if (x > 0)
    x--;
  else
    x++;
}
\end{lstlisting}
\caption{Taken from \cite{DBLP:conf/tacas/CookSZ13}}
 \label{fig:motivation.g}
\end{subfigure}

&

\begin{subfigure}[b]{0.45\textwidth}
\begin{lstlisting}
y = 1;

while (x > 0) {
  x = x - y;
}
\end{lstlisting}
\caption{}
 \label{fig:motivation.h}
\end{subfigure} 
&
\begin{subfigure}[b]{0.45\textwidth}
\begin{lstlisting}
while (x>0 && y>0 && z>0){
  if (y > x) {
    y = z;
    x = nondet();
    z = x - 1;
  } else {
    z = z - 1;
    x = nondet();
    y = x - 1;
  }
}
\end{lstlisting}
\caption{Taken from~\cite{BA:mcs}}
 \label{fig:motivation.i}
\end{subfigure} 
\end{tabular}
\caption{Motivational examples, mostly taken from the literature.\label{fig:motivation}}
\end{figure*}

\section{Preliminaries}

Given a program, we first formalise its termination argument as a ranking
function (Section~\ref{sec:ranking.functions}).  Subsequently, we discuss
bit-vector semantics and illustrate differences between machine arithmetic
and integer arithmetic that show that the abstraction of bit-vectors to
mathematical integers is unsound (Section~\ref{sec:machine.arith}).

\subsection{Termination and Ranking Functions} \label{sec:ranking.functions}

A program $P$ is represented as a transition system with state space $X$ and
transition relation $T \subseteq X \times X$.  For a state
$x \in X$ with $T(x,x')$ we say $x'$ is a successor of $x$ under $T$.

\begin{definition}[Unconditional termination]
A program is said to be \emph{unconditionally terminating} if
there is no infinite sequence of states $x_1, x_2, \ldots \in X$ with
$\forall i.~T(x_i, x_{i+1})$.
\end{definition}

We can prove that the program is unconditionally terminating by
finding a ranking function for its transition relation.
\begin{definition}[Ranking function]
A function ${R:X\to Y}$ is a \emph{ranking function} for the
transition relation $T$ if $Y$ is a well-founded set with order $>$ and 
$R$ is injective and monotonically decreasing with respect to $T$.  That is
to say:
$$\forall x, x' \in X. T(x, x') \Rightarrow R(x) > R(x')$$
\end{definition}

\begin{definition}[Linear function]
A \emph{linear function} $f: X \to Y$ 
with $\dim(X) = n$ and $\dim(Y) = m$ is of the form: $$f(\vec{x}) = M\vec{x}$$ where
$M$ is an $n \times m$ matrix.
\end{definition}

In the case that $\dim(Y) = 1$, this reduces to the inner product
$$f(\vec{x}) = \vec{\lambda} \cdotp \vec{x} + c \;.$$

\begin{definition}[Lexicographic ranking function]
For $Y = Z^m$, we say that a ranking function $R: X \to Y$ is \emph{lexicographic}
if it maps each state in $X$ to a tuple of values such that the loop transition leads to a decrease with
respect to the lexicographic ordering for this tuple.
The total order imposed on $Y$ is the lexicographic ordering
induced on tuples of $Z$'s.  So for $y = (z_1, \ldots, z_m)$ and
$y' = (z'_1, \ldots, z'_m)$:
\[
 y > y' \iff \exists i \leq m . z_i > z'_i \wedge \forall j < i . z_j = z'_j
\]

\end{definition}

We note that some termination arguments require lexicographic ranking functions, or
alternatively, ranking functions whose co-domain is a countable ordinal, rather than just $\mathbb{N}$.

\subsection{Machine Arithmetic Vs.~Peano Arithmetic} \label{sec:machine.arith} 

Physical computers have bounded storage, which means they are unable to
perform calculations on mathematical integers.  
They do their arithmetic over
fixed-width binary words, otherwise known as bit-vectors.  For the remainder
of this section, we will say that the bit-vectors we are working with are
$k$-bits wide, which means that each word can hold one of $2^k$ bit
patterns.  Typical values for $k$ are 32 and 64.

Machine words can be interpreted as ``signed'' or ``unsigned'' values. 
Signed values can be negative, while unsigned values cannot.  The encoding
for signed values is two's complement, where the most significant bit
$b_{k-1}$ of the word is a ``sign'' bit, whose weight is $-(2^k - 1)$ rather
than $2^k - 1$.  Two's complement representation has the property that
$\forall x .  -x = (\mathord{\sim} x) + 1$, where $\mathord{\sim}(\bullet)$
is bitwise negation.  Two's complement also has the property that addition,
multiplication and subtraction are defined identically for unsigned and
signed numbers.

Bit-vector arithmetic is performed modulo $2^k$, which is the source of many
of the differences between machine arithmetic and Peano
arithmetic\footnote{ISO C requires that unsigned arithmetic is performed
modulo $2^k$, whereas the overflow case is undefined for signed arithmetic. 
In practice, the undefined behaviour is implemented just as if the
arithmetic had been unsigned.}.  To give an example, $(2^k - 1) + 1 \equiv 0
\pmod {2^k}$ provides a counterexample to the statement $\forall x. 
x + 1 > x$, which is a theorem of Peano arithmetic but not of modular
arithmetic.  When an arithmetic operation has a result greater than $2^k$,
it is said to ``overflow''.  If an operation does not overflow, its
machine-arithmetic result is the same as the result of the same operation
performed on integers.

The final source of disagreement between integer arithmetic and bit-vector
arithmetic stems from width conversions.  Many programming languages allow
numeric variables of different types, which can be represented using words
of different widths.  In C, a \texttt{short} might occupy 16 bits, while an
\texttt{int} might occupy 32 bits.  When a $k$-bit variable is assigned to a
$j$-bit variable with $j < k$, the result is truncated $\mathrm{mod}~2^j$.  For
example, if $x$ is a 32-bit variable and $y$ is a 16-bit variable, $y$ will
hold the value $0$ after the following code is executed:
\begin{lstlisting}
x = 65536;
y = x;
\end{lstlisting}


As well as machine arithmetic differing from Peano arithmetic on the
operators they have in common, computers have several ``bitwise'' operations
that are not taken as primitive in the theory of integers.  These operations
include the Boolean operators \texttt{and, or, not, xor} applied to each
element of the bit-vector.  Computer programs often make use of these
operators, which are nonlinear when interpreted in the standard model of
Peano arithmetic\footnote{Some of these operators can be seen as
linear in a different algebraic structure, e.g.~\texttt{xor} corresponds to
addition in the Galois field $\mathrm{GF}(2^k)$.}.

\section{Termination as Second-Order Satisfaction} \label{sec:second.order}

The problem of program verification can be reduced to the problem of finding
solutions to a second-order
constraint~\cite{DBLP:conf/pldi/GrebenshchikovLPR12,DBLP:conf/pldi/GulwaniSV08}. 
Our intention is to apply this approach to termination analysis.  In this
section we show how several variations of both the termination and the
non-termination problem can be uniformly defined in second-order logic.

Due to its expressiveness, second-order logic is very difficult to reason
in, with many second-order theories becoming undecidable even when the
corresponding first-order theory is decidable.
In \cite{kalashnikov}, we have identified and built a solver for a fragment of second-order logic with restricted quantification, 
which we call second-order SAT (see Definition~\ref{def:2sat}).

\begin{definition}[Second-Order SAT]
\label{def:2sat}
 \[
  \exists S_1 \ldots S_m . Q_1 x_1 \ldots Q_n x_n . \sigma
 \]
 Where the $S_i$'s range over predicates,
the $Q_i$'s are either $\exists$ or $\forall$,
the $x_i$'s range over boolean values,
 and $\sigma$ is a quantifier-free propositional formula
 whose free variables are the $x_i$'s.  
Each $S_i$ has an associated arity $\mathrm{ar}(S_i)$
 and $S_i \subseteq \mathbb{B}^{\mathrm{ar}(S_i)}$.  Note that 
$Q_1 x_1 \ldots Q_n x_n . \sigma$
 is an instance of first-order propositional SAT, i.e. QBF.
\end{definition}


%
We note that by existentially quantifying over Skolem functions, formulae with arbitrary
first-order quantification can be brought into the synthesis fragment~\cite{hol-book}, so the
fragment is semantically less restrictive than it looks.

In the rest of this section, we show that second-order SAT 
is expressive enough to encode both termination and non-termination. 

\subsection{An Isolated, Simple Loop}

We will begin our discussion by showing how to encode in second-order SAT the
\mbox{(non-)termination} of a program consisting of a single loop with no nesting.
For the time being, a loop $L(G, T)$ is defined by its guard $G$ and body $T$
such that states $x$ satisfying the loop's guard are given by the
predicate $G(x)$.  The body of the loop is encoded as the transition
relation $T(x, x')$, meaning that state $x'$ is reachable from state $x$ via
a single iteration of the loop body.  For example, the loop in
Figure~\ref{fig:motivation.a} is encoded as:
\begin{align*}
G(x) & = \{ x \mid x>0 \} \\
T(x,x') &= \{ \langle x, x' \rangle \mid x' = (x - 1) \, \& \, x \}
\end{align*}
We will abbreviate this with the notation:
\begin{align*}
G(x) & \triangleq x > 0 \\
T(x, x') & \triangleq x' = (x - 1) \, \& \, x
\end{align*}

\begin{figure*}
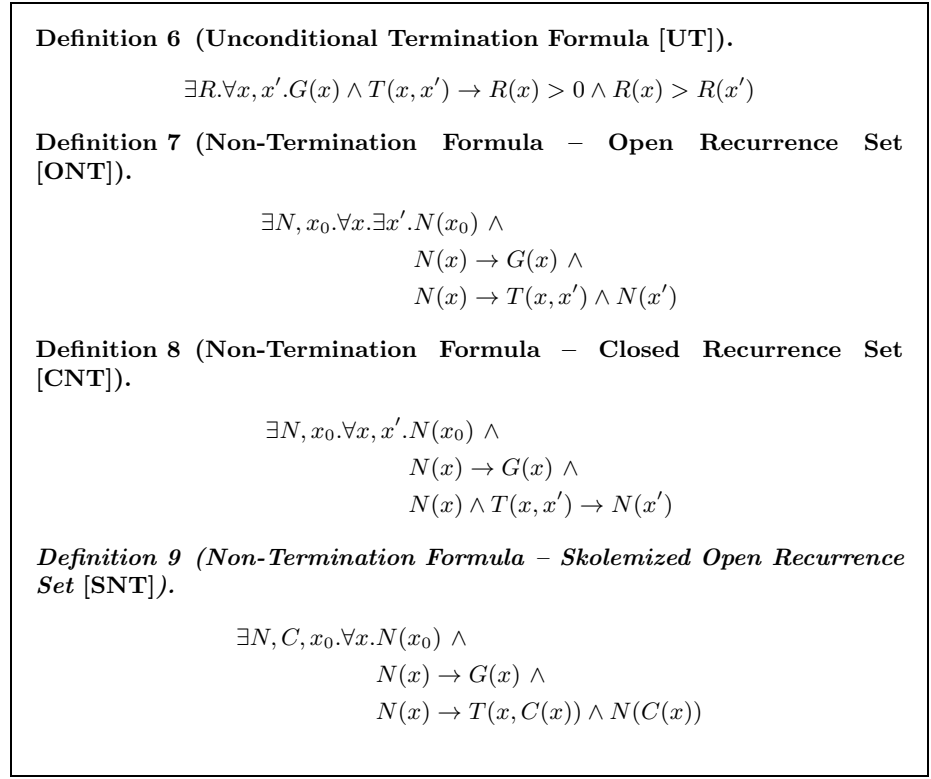

\begin{framed}
\begin{definition}[Unconditional Termination Formula {\bf [UT]}]
\label{def:UT}
\begin{align*}
 \exists R . \forall x, x' . & G(x) \wedge T(x, x') \rightarrow R(x) > 0 \wedge R(x) > R(x')
\end{align*}
\end{definition}

\begin{definition}[Non-Termination Formula -- Open Recurrence Set  {\bf [ONT]}]
\label{def:ont}
 \begin{align*}
  \exists N, x_0 . \forall x . \exists x' . & N(x_0) ~\wedge \\ &  N(x) \rightarrow G(x) ~ \wedge \\
							& N(x) \rightarrow T(x, x') \wedge N(x') 
 \end{align*}
\end{definition}

\begin{definition}[Non-Termination Formula -- Closed Recurrence Set {\bf [CNT]}]
\label{def:cnt}
 \begin{align*}
  \exists N, x_0 . \forall x, x' . & N(x_0) ~ \wedge \\ & N(x) \rightarrow G(x) ~ \wedge \\
							& N(x) \wedge T(x, x') \rightarrow N(x') 
 \end{align*}

\begin{definition}[Non-Termination Formula -- Skolemized Open Recurrence Set  {\bf [SNT]}]
\label{def:snt}
 \begin{align*}
  \exists N, C, x_0 . \forall x . & N(x_0) ~\wedge \\ &  N(x) \rightarrow G(x) ~ \wedge \\
							& N(x) \rightarrow T(x, C(x)) \wedge N(C(x))
 \end{align*}
\end{definition}

\end{definition}


\end{framed}
\caption{Formulae encoding the termination and non-termination of a single loop} \label{fig:single_loop}
\end{figure*}

\noindent {\bf Unconditional termination.}
We say that a loop $L(G, T)$ is unconditionally terminating iff it eventually
terminates regardless of the state it starts in. To prove unconditional termination, it suffices
to find a ranking function for \mbox{$T \cap (G \times X)$}, i.e.~$T$ restricted to states satisfying the loop's guard.


\begin{theorem}
\label{thm:ut}
 The loop $L(G, T)$ terminates from every start state iff formula {\bf [UT]} (Definition~\ref{def:UT}, Figure~\ref{fig:single_loop}) is satisfiable.
\end{theorem}

%

As the existence of a ranking function is equivalent to the satisfiability
of the formula {\bf [UT]}, a satisfiability witness is a ranking
function and thus a proof of $L$'s unconditional termination.

Returning to the program from Figure~\ref{fig:motivation.a}, we
can see that the corresponding second-order SAT formula {\bf [UT]} is satisfiable,
as witnessed by the function $R(x) = x$.  Thus, $R(x) = x$ constitutes a
proof that the program in Figure~\ref{fig:motivation.a} is unconditionally
terminating.

Note that different formulations for unconditional termination are possible. 
We are aware of a proof rule based on transition invariants, i.e.~supersets
of the transition relation's transitive
closure~\cite{DBLP:conf/pldi/GrebenshchikovLPR12}.  This formulation assumes
that the second-order logic has a primitive predicate for disjunctive
well-foundedness.  By contrast, our formulation in Definition~\ref{def:UT}
does not use a primitive disjunctive well-foundedness predicate.  \\

\noindent{\bf Non-termination.}
Dually to termination, we might want to consider the non-termination of a loop.  If a loop terminates,
we can prove this by finding a ranking function 
witnessing the satisfiability of formula {\bf[UT]}.  What then would a proof of non-termination look like?

Since our program's state space is finite, a transition relation
induces an infinite execution iff some state is visited infinitely
often, or equivalently $ \exists x . T^+(x, x)$.
Deciding satisfiability of this formula directly would require a logic
that includes a transitive closure operator, $\bullet^+$.  Rather than
introduce such an operator, we will characterise non-termination
using the second-order SAT formula {\bf [ONT]} (Definition~\ref{def:ont}, Figure~\ref{fig:single_loop})
encoding the existence of an \emph{(open) recurrence set}, i.e.~a nonempty 
set of states $N$ such that for each $s \in N$ there
exists a transition to some $s' \in N$ \cite{DBLP:conf/popl/GuptaHMRX08}.

\begin{theorem}
\label{thm:ont}
 The loop $L(G, T)$ has an infinite execution iff formula {\bf [ONT]} (Definition~\ref{def:ont}) is satisfiable.
\end{theorem}

%

If this formula is satisfiable, $N$ is an open recurrence set for $L$, which proves
$L$'s non-termination. The issue with this formula is the additional level of quantifier alternation as compared to second-order SAT
(it is an $\exists \forall \exists$ formula).  To eliminate the innermost existential quantifier,
we introduce a Skolem function $C$ that chooses the successor $x'$, which we then existentially quantify over.
This results in formula {\bf [SNT]} (Definition~\ref{def:snt}, Figure~\ref{fig:single_loop}).

\begin{theorem}
 \label{thm:snt}
 Formula {\bf [ONT]} (Definition~\ref{def:ont}) and formula {\bf [SNT]} (Definition~\ref{def:snt}) are equisatisfiable.
\end{theorem}

This extra second-order term introduces some complexity to the formula, which
we can avoid if the transition relation $T$ is deterministic.
\begin{definition}[Determinism]
 A relation $T$ is deterministic iff each state $x$ has exactly one successor under $T$:
 $$\forall x . \exists x' . T(x, x') \wedge \forall x'' . T(x, x'') \rightarrow x'' = x'$$
\end{definition}
In order to describe a deterministic \emph{program} in a way that still allows us
to sensibly talk about termination, we assume the existence of a special sink
state $s$ with no outgoing transitions and such that $\lnot G(s)$ for any
of the loop guards $G$.  The program is deterministic if its transition
relation is deterministic for all states except $s$.

When analysing a deterministic loop, we can make use of the notion of a \emph{closed recurrence set}
introduced by Chen et al.~in~\cite{DBLP:conf/tacas/ChenCFNO14}:  for each
state in the recurrence set $N$, \emph{all} of its successors must be in $N$.
The existence of a closed recurrence set is equivalent to the satisfiability
of formula {\bf [CNT]} in Definition~\ref{def:cnt}, which is already in second-order SAT without needing Skolemization.

We note that if $T$ is deterministic, every open recurrence set is also a
closed recurrence set (since each state has at most one successor).  Thus,
the non-termination problem for deterministic transition systems is
equivalent to the satisfiability of formula {\bf [CNT]} from Figure~\ref{fig:single_loop}.

\begin{theorem}
\label{thm:cnt}
 If $T$ is deterministic,
formula {\bf [ONT]} (Definition~\ref{def:ont}) and formula {\bf [CNT]} (Definition~\ref{def:cnt}) are equisatisfiable.
\end{theorem}

%


So if our transition relation is deterministic, we can say, without
loss of generality, that non-termination of the loop is equivalent
to the existence of a closed recurrence set.  However if $T$ is
non-deterministic, it may be that there is an open recurrence
set but not closed recurrence set.  To see this, consider the following
loop:
\begin{lstlisting}[language=C]
while(x != 0) {
  y = nondet();
  x = x-y;
}
\end{lstlisting}

It is clear that this loop has many non-terminating executions,
e.g. the execution where \lstinline!nondet()! always returns 0.
However each state has a successor
that exits the loop, i.e.~when \lstinline|nondet()| returns
the value currently stored in \lstinline|x|.  So this loop
has an open recurrence set, but no closed recurrence set
and hence we cannot give a proof of its non-termination
with {\bf [CNT]} and instead must use {\bf [SNT]}.

%

\subsection{An Isolated, Nested Loop}
\noindent {\bf Termination.} If a loop $L(G, T)$ has another loop $L'(G', T')$ nested inside it, we cannot directly use {\bf [UT]}
to express the termination of $L$.  This is because the single-step transition relation $T$ must
include the transitive closure of the inner loop $T'^*$, and we do not have a transitive closure
operator in our logic.  Therefore to encode the termination of $L$, we construct an over-approximation
$T_o \supseteq T$ and use this in formula {\bf [UT]} to specify a ranking function.
Rather than explicitly construct $T_o$ using, for example, abstract interpretation, we add constraints to
our formula that encode the fact that $T_o$ is an over-approximation of $T$, and that it is
precise enough to show that $R$ is a ranking function.


As the generation of such constraints is standard and covered by several other works \cite{DBLP:conf/pldi/GrebenshchikovLPR12,DBLP:conf/pldi/GulwaniSV08}, 
we will not provide the full algorithm, but rather illustrate it through the example in Figure~\ref{fig:environment-model}.
Full details of this construction appear in the extended version of this paper.
For the current example, the termination formula is given on the right side of 
Figure~\ref{fig:environment-model}:
$T_o$ is a summary of $L_1$ that over-approximates its transition relation;
$R_1$ and $R_2$ are ranking functions for $L_1$ and $L_2$, respectively.

\begin{figure*}
\begin{framed}
\begin{minipage}{.2\textwidth}
\begin{lstlisting}[mathescape=true,basicstyle=\tiny]
$L_1:$
while (i<n){
  j = 0;

$L_2:$
  while (j$\leq$i){
    j = j + 1;
  }

  i = i + 1;
}
\end{lstlisting}
\end{minipage}
\vline
\resizebox{.75\textwidth}{!}{
\begin{minipage}{.85\textwidth}
 \begin{align*}
 \exists T_o, R_1, R_2 . \forall i, j, n, i', j', n' . \\
   i < n & \rightarrow T_o(\langle i, j, n \rangle , \langle i, 0, n \rangle) ~ \wedge \\
   j \leq i \wedge T_o(\langle i', j', n' \rangle, \langle i, j, n\rangle ) & \rightarrow R_2(i, j, n) > 0 ~ \wedge \\
     & R_2(i, j, n) > R_2(i, j+1, n) ~ \wedge \\
     & T_o(\langle i', j', n' \rangle, \langle i, j+1, n \rangle) ~ \wedge \\
   i < n \wedge S(\langle i, j, n \rangle, \langle i', j', n' \rangle) \wedge j' > i' & \rightarrow R_1(i, j, n) > 0 ~ \wedge \\
                                         & R_1(i, j, n) > R_1(i+1, j, n)
   \end{align*}
\end{minipage}
}
\end{framed}
\caption{A program with nested loops and its termination formula\label{fig:environment-model}}
\end{figure*}




\begin{figure*}
 \begin{framed}

\begin{definition}[Conditional Termination Formula {\bf [CT]}]
\label{def:ct}
 \begin{align*}
  \exists R, W . \forall x, x' . & I(x) \wedge G(x) \rightarrow W(x) ~ \wedge \\
                                 & G(x) \wedge W(x) \wedge T(x, x') \rightarrow W(x') \wedge R(x) > 0 
  \wedge R(x) > R(x')
 \end{align*}
\end{definition}

 \end{framed}
\caption{Formula encoding conditional termination of a loop} \label{fig:conditional_termination}
\end{figure*}

\noindent {\bf Non-Termination.}
Dually to termination, when proving non-termination, we need to
under-approximate the loop's body and apply formula {\bf [CNT]}.
%
%
%
%
Under-approximating the inner loop can be done with a nested existential quantifier, resulting in
$\exists \forall \exists$ alternation, which we could eliminate with Skolemization.  
However, we observe that
unlike a ranking function,
the defining property of a recurrence set is \emph{non relational} -- if we
end up in the recurrence set, we do not care exactly where we came from as
long as we know that it was also somewhere in the recurrence set.  
This allows us to cast non-termination of nested loops as the formula shown in
Figure~\ref{fig:nonterm-nested}, which does not use a Skolem function.

If the formula on the right-hand side of the figure is satisfiable, then
$L_1$ is non-terminating, as witnessed by the recurrence set $N_1$ and the
initial state $x_0$ in which the program begins executing.  There are two
possible scenarios for $L_2$'s termination:
\begin{itemize}
\item If $L_2$ is terminating, then $N_2$ is an inductive invariant that
reestablished $N_1$ after $L_2$ stops executing: $\lnot G_2(x) \wedge N_2(x)
\wedge P_2(x,x') \rightarrow N_1(x') $.
\item If $L_2$ is non-terminating, then $N_2 \wedge G_2$ is its recurrence set.
\end{itemize}





\begin{figure*}
\begin{framed}
\begin{minipage}{0.2\textwidth}
\begin{lstlisting}[mathescape=true,basicstyle=\tiny]
$L_1$:
while ($G_1$) {
  $P_1$;

$L_2$:
  while ($G_2$) {
    $B_2$;
  }

  $P_2$;
}
\end{lstlisting}
\end{minipage}
\vline
\begin{minipage}{0.8\textwidth}
\begin{align*}
 \exists N_1, N_2, x_0 . \forall x, x' . \\
  N_1(x_0) & \, \wedge \\
  N_1(x) & \rightarrow G_1(x) \, \wedge \\
  N_1(x) \wedge P_1(x,x') & \rightarrow N_2(x') \, \wedge \\
  G_2(x) \wedge N_2(x) \wedge B_2(x,x') & \rightarrow N_2(x') \, \wedge \\
  \lnot G_2(x) \wedge N_2(x) \wedge P_2(x,x') & \rightarrow N_1(x') 
\end{align*}
\end{minipage}
\end{framed}

\caption{Formula encoding non-termination of nested loops \label{fig:nonterm-nested}}
\end{figure*}

\subsection{Composing a Loop with the Rest of the Program} \label{sec:env}

Sometimes the termination behaviour of a loop depends on the rest of the program.  That is to say,
the loop may not terminate if started in some particular state, but that state is
not actually reachable on entry to the loop.  The program as a whole
terminates, but if the loop were considered in isolation we would not be able to prove that
it terminates. We must therefore encode a loop's interaction with the rest of the program 
in order to do a sound termination analysis.\\

Let us assume that we have done some preprocessing of our program which has identified
loops, straight line code blocks and the control flow between these.  In particular,
the control flow analysis has determined which order these code blocks execute in,
and the nesting structure of the loops.

\noindent {\bf Conditional termination.}
Given a loop $L(G,T)$, if $L$'s termination depends on the state it begins
executing in, we say that $L$ is \emph{conditionally terminating}.
The information we require of the rest of the program is a predicate $I$ which
over-approximates the set of states that $L$ may begin executing in.
That is to say, for each state $x$ that is reachable on entry to $L$,
we have $I(x)$.

\begin{theorem}
\label{thm:ct}
 The loop $L(G, T)$ terminates when started in any state satisfying $I(x)$ iff formula {\bf [CT]}
 (Definition~\ref{def:ct}, Figure~\ref{fig:conditional_termination}) is satisfiable.
\end{theorem}
%

If formula {\bf [CT]} is satisfiable, two witnesses are returned:
\begin{itemize}
\item $W$ is an inductive invariant of $L$ that is established by the initial states $I$ if the loop
guard $G$ is met.
\item $R$ is a ranking function for $L$ as restricted by $W$ -- that is to say, $R$ need only
be well founded on those states satisfying $W \wedge G$.  Since $W$ is an inductive invariant of $L$,
$R$ is strong enough to show that $L$ terminates from any of its initial states.
\end{itemize}

$W$ is called a \emph{supporting invariant} for $L$ and $R$ proves termination relative to $W$.
We require that $I \wedge G$ is strong enough to establish the base case of $W$'s inductiveness.

Conditional termination is illustrated by the program in Figure~\ref{fig:motivation.h},
which is encoded as:
\begin{align*}
            I(\langle x, y \rangle) & \triangleq y = 1 \\
            G(\langle x, y \rangle) & \triangleq x > 0 \\
            T(\langle x, y \rangle, \langle x', y' \rangle) & \triangleq x' = x - y \wedge y' = y 
\end{align*}
If the initial states $I$ are ignored, this loop cannot be shown to terminate, since any state with $y = 0$ and $x > 0$
would lead to a non-terminating execution.

However, formula {\bf [CT]} is satisfiable, as witnessed by:
\begin{align*}
R(\langle x,y\rangle) & = x\\
W(\langle x, y \rangle ) & \triangleq y  = 1
\end{align*}

This constitutes a proof that the program as a whole terminates, since the loop always begins
executing in a state that guarantees its termination.\\

\subsection{Generalised Termination and Non-termination Formula}

At this point, we know how to construct two formulae for a loop $L$: one
that is satisfiable iff $L$ is terminating and another that is satisfiable
iff it is non-terminating.  We will call these formulae $\phi$ and $\psi$,
respectively:
\begin{align*}
 \exists P_T . \forall x, x' . \phi(P_T, x, x') \\
 \exists P_N . \forall x . \psi(P_N, x)
\end{align*}
We can combine these:
\begin{align*}
 (\exists P_T . \forall x, x'. \phi(P_T, x, x')) \vee (\exists P_N . \forall x.\, \psi(P_N, x))
\end{align*}
Which simplifies to:
\begin{definition}[Generalised Termination Formula {\bf [GT]}]
\label{def:general-termination}
\begin{align*}
 \exists P_T, P_N. \forall x, x', y.\, \phi(P_T, x, x') \vee \psi(P_N, y)
\end{align*}
\end{definition}

Since $L$ either terminates or does not terminate, this formula is a tautology in second-order SAT.
A solution to the formula would include witnesses $P_N$ and $P_T$, which are putative proofs of non-termination
and termination respectively.  Exactly one of these will be a genuine proof, so we can check
first one and then the other.

%
%



\subsection{Solving the Second-Order SAT Formula}
In order to solve the second-order generalised formula {\bf [GT]}, we use the solver described in \cite{kalashnikov}.
%
%
%
For any satisfiable formula, the solver is guaranteed to find a satisfying assignment to
all the second-order variables. 

In the context of our termination analysis, such a satisfying assignment returned by the solver 
represents either a proof of termination or non-termination, and
takes the form of an imperative program
written in the language $\mathcal{L}$.  An $\mathcal{L}$-program is a list
of instructions, each of which matches one of the patterns shown in
Figure~\ref{fig:l-language}.  An instruction has an opcode (such as
\verb|add| for addition) and one or more operands.  An operand is either a
constant, one of the program's inputs or the result of a previous
instruction.  The $\mathcal{L}$ language has various arithmetic and logical
operations, as well as basic branching in the form of the \verb|ite|
(if-then-else) instruction.

\begin{figure}
\begin{center}
{\small

\setlength{\tabcolsep}{14pt}
Integer arithmetic instructions:

\begin{tabular}{llll}
 \verb|add a b| & \verb|sub a b| & \verb|mul a b| & \verb|div a b| \\
 \verb|neg a| &   \verb|mod a b| & \verb|min a b| & \verb|max a b|
\end{tabular}

\medskip

Bitwise logical and shift instructions:

\begin{tabular}{lll}
 \verb|and  a b| & \verb|or   a b| & \verb|xor a b| \\
 \verb|lshr a b| & \verb|ashr a b| & \verb|not a|
\end{tabular}

\medskip

Unsigned and signed comparison instructions:

\begin{tabular}{lll}
 \verb|le  a b| & \verb|lt  a b| & \verb|sle  a b| \\
 \verb|slt a b| & \verb|eq  a b| & \verb|neq  a b| \\
\end{tabular}

\medskip

Miscellaneous logical instructions:

\begin{tabular}{lll}
 \verb|implies a b| & \verb|ite a b c| &  \\
\end{tabular}

\medskip

\setlength{\tabcolsep}{12pt}

Floating-point arithmetic:

\begin{tabular}{llll}
 \verb|fadd a b| & \verb|fsub a b| & \verb|fmul a b| & \verb|fdiv a b|
\end{tabular}

}
\end{center}

 \caption{The language $\mathcal{L}$}
 \label{fig:l-language}
\end{figure}

\section{Soundness, Completeness and Complexity}
In this section, we show that $\mathcal{L}$ is expressive enough
to capture  (non-)termination proofs for every bit-vector program.
By using this result, 
we then show that our analysis terminates with a valid proof for every input program.

\begin{lemma}
 \label{lem:l-func}
 Every function $f: X \to Y$ for finite $X$ and $Y$ is computable by a finite $\mathcal{L}$-program.
\end{lemma}

\begin{proof}
Without loss of generality, let $X = Y = \mathbb{N}_b^k$ the set of
$k$-tuples of natural numbers less than $b$.
A very inefficient construction which computes the first coordinate
of the output $y$ is:
\begin{verbatim}
t1 = f(0)
t2 = v1 == 1
t3 = ITE(t2, f(1), t1)
t4 = v1 == 2
t5 = ITE(t4, f(2), t3)
...
\end{verbatim}
Where the \verb|f(n)| are literal constants that are to appear in the program text.
This program is of length $2b - 1$, and so all $k$ co-ordinates of the output $y$
are computed by a program of size at most $2bk - k$.
\end{proof}

\begin{corollary}
 Every finite subset $A \subseteq B$ is computable by a finite $\mathcal{L}$-program
 by setting $X = B, Y = 2$ in Lemma~\ref{lem:l-func} and taking the
 resulting function to be the characteristic function of $A$.
\end{corollary}

\begin{theorem}
\label{thm:l-term}
 Every terminating bit-vector program has a ranking function that is expressible in $\mathcal{L}$.
\end{theorem}

\begin{proof}
Let $v_1, \ldots, v_k$ be the variables of the program $P$ under analysis,
and let each be $b$ bits wide.  Its state space $\mathcal{S}$ is then of
size $2^{bk}$.  A~ranking function $R: \mathcal{S} \to \mathcal{D}$ for $P$
exists iff $P$ terminates.  Without loss of generality,
$\mathcal{D}$ is a well-founded total order.  Since $R$ is injective, we have that $\|
\mathcal{D} \| \geq \| \mathcal{S} \|$.  If $\| \mathcal{D} \| > \|
\mathcal{S} \|$, we can construct a function $R': \mathcal{S} \to
\mathcal{D'}$ with $ \| \mathcal{D'} \| = \| \mathcal{S} \|$ by just setting
$R' = R|_\mathcal{S}$, i.e.~$R'$ is just the restriction of $R$ to
$\mathcal{S}$.  Since $\mathcal{S}$ already comes equipped with a natural
well ordering we can also construct $R'' = \iota \circ R'$ where $\iota:
\mathcal{D'} \to \mathcal{S}$ is the unique order isomorphism from
$\mathcal{D'}$ to $\mathcal{S}$.  So assuming that $P$ terminates, there is
some ranking function $R''$ that is just a permutation of $\mathcal{S}$.  If
the number of variables $k > 1$ then in general the ranking function will be
lexicographic with dimension $\leq k$ and each co-ordinate of the output
being a single $b$-bit value.

Then by Lemma~\ref{lem:l-func} with $X = Y = \mathcal{S}$, there exists
a finite $\mathcal{L}$-program computing $R''$.
\end{proof}

\begin{theorem}
\label{thm:l-nonterm}
 Every non-terminating bit-vector program has a non-termination proof expressible in $\mathcal{L}$.
\end{theorem}

\begin{proof}
 A proof of non-termination is a triple $\langle N, C, x_0 \rangle$ where
 \mbox{$N \subseteq \mathcal{S}$} is a (finite) recurrence set and
 $C : \mathcal{S} \to \mathcal{S}$ is a Skolem function choosing
 a successor for each $x \in N$.  $\mathcal{S}$ is finite, so by Lemma~\ref{lem:l-func} both
 $N$ and $C$ are computed by finite $\mathcal{L}$-programs and $x_0$ is just a ground term.
\end{proof}

\begin{theorem}
 \label{thm:generalised-sat}
 The generalised termination formula {\bf [GT]} for any loop $L$ is a tautology
 when $P_N$ and $P_T$ range over \mbox{$\mathcal{L}$-computable} functions.
\end{theorem}

\begin{proof}
 For any $P, P', \sigma, \sigma$, if $P \models \sigma$ then $(P, P') \models \sigma \vee \sigma'$.

 By Theorem~\ref{thm:l-term}, if $L$ terminates then there exists a termination proof $P_T$ expressible
 in $\mathcal{L}$.  Since $\phi$ is an instance of {\bf [CT]}, $P_T \models \phi$ (Theorem~\ref{thm:ct}) and
 for any $P_N$, $(P_T, P_N) \models \phi \vee \psi$.

 Similarly if $L$ does not terminate for some input, by Theorem~\ref{thm:l-nonterm} there is a non-termination
 proof $P_N$ expressible in $\mathcal{L}$.  Formula $\psi$ is an instance of {\bf [SNT]} and so $P_N \models \psi$
 (Theorem~\ref{thm:snt}), hence for any $P_T$, $(P_T, P_N) \models \phi \vee \psi$.

 So in either case ($L$ terminates or does not), there is a witness in $\mathcal{L}$ satisfying
 $\phi \vee \psi$, which is an instance of {\bf [GT]}.
\end{proof}

\begin{theorem}
Our termination analysis 
is sound and complete -- it terminates for all input loops $L$ with
 a correct termination verdict.
\end{theorem}

\begin{proof}
 By Theorem~\ref{thm:generalised-sat}, the specification \lstinline!spec! is satisfiable.
 In~\cite{kalashnikov}, we show that the second-order SAT solver is semi-complete, and so is guaranteed to find a satisfying
 assignment for \lstinline!spec!. 
 If $L$ terminates then $P_T$ is a termination proof (Theorem~\ref{thm:ct}),
 otherwise $P_N$ is a non-termination proof (Theorem~\ref{thm:snt}).  Exactly one of these purported proofs
 will be valid, and since we can check each proof with a single call to a SAT solver we simply
 test both and discard the one that is invalid.
\end{proof}

\section{Experiments}

To evaluate our algorithm, we implemented a tool that generates a
termination specification from a C program and calls the second-order SAT solver in \cite{kalashnikov} to obtain a proof.
We ran the resulting termination prover,
named {\sc Juggernaut}, on 47
benchmarks taken from the literature and
SV-COMP'15~\cite{svcomp15}.
We omitted exactly those SVCOMP'15
benchmarks that made use of arrays or recursion. We do not have arrays in
our logic and we had not implemented recursion in our frontend (although the
latter can be syntactically rewritten to our input format).

To provide a comparison point, we also ran {\sc ARMC}~\cite{armc-website} on
the same benchmarks.  Each tool was given a time limit of 180\,s, and was
run on an unloaded 8-core 3.07\,GHz Xeon X5667 with 50\,GB of RAM.  The
results of these experiments are given in Figure~\ref{fig:experiments}.

It should be noted that the comparison here is imperfect, since {\sc ARMC}
is solving a different problem -- it checks whether the program
under analysis would terminate if run with unbounded integer variables,
while we are checking whether the program terminates with bit-vector
variables.  This means that {\sc ARMC}'s verdict differs from ours
in 3 cases (due to the differences between integer and bit-vector
semantics).  There are a further 7 cases where our tool is able to find a
proof and {\sc ARMC} cannot, which we believe is due to our more expressive
proof language.  In 3 cases, {\sc ARMC} times out while our tool is able to
find a termination proof.  Of these, 2 cases have nested loops and the
third has an infinite number of terminating lassos.  This is not a problem for us,
but can be difficult for provers that enumerate lassos.

On the other hand, {\sc ARMC} is \emph{much} faster than our tool.  While
this difference can partly be explained by much more engineering time being
invested in {\sc ARMC}, we feel that the difference is probably inherent to
the difference in the two approaches -- our solver is more general
than {\sc ARMC}, in that it provides a complete proof system for both
termination and non-termination.  This comes at the cost of efficiency:
{\sc Juggernaut} is slow, but unstoppable.


Of the 47 benchmarks, 2 use nonlinear operations in the program (loop6 and loop11),
and 5 have nested loops (svcomp6, svcomp12, svcomp18, svcomp40, svcomp41).
{\sc Juggernaut} handles the nonlinear cases correctly and rapidly.
It solves 4 of the 5 nested loops in less than 30\,s, but times out on the 5th.


In conclusion,
these experiments confirm our conjecture that second-order SAT can be used
effectively to prove termination and non-termination.  In particular,
for programs with nested loops, nonlinear arithmetic and complex
termination arguments, the versatility given by a general purpose solver 
is very valuable.

\begin{figure}
\centering
\small
\begin{tabular}{|l|@{}c@{}||@{}c@{}|r||@{}c@{}|r|}
\hline
          &             & \multicolumn{2}{|c||}{\sc ARMC} & \multicolumn{2}{|c|}{\sc Juggernaut} \\
Benchmark & \,Expected\, & \,Verdict\, & Time & \,Verdict\, & Time \\
    \hline
    \hline
\input{table}
    \hline
\end{tabular}

Key: \tick = terminating, \xmark = non-terminating, ? = unknown (tool terminated with an inconclusive verdict).

\caption{Experimental results\label{fig:experiments}}
 \end{figure}

\section{Conclusions and Related Work}

There has been substantial prior work on automated program termination
analysis.  Figure~\ref{fig:handletable} summarises the related work with
respect to the assumptions they make about programs and ranking functions. 
Most of the techniques are specialised in the synthesis of linear ranking
functions for linear programs over integers (or
rationals)~\cite{DBLP:conf/pldi/CookPR06,DBLP:conf/cav/LeeWY12,DBLP:conf/popl/Ben-AmramG13,DBLP:conf/vmcai/P04,DBLP:conf/atva/HeizmannHLP13,DBLP:conf/cav/BradleyMS05,DBLP:conf/tacas/CookSZ13,DBLP:conf/cav/KroeningSTW10}.  Among them, Lee et
al.~make use of transition predicate abstraction, algorithmic learning, and
decision procedures~\cite{DBLP:conf/cav/LeeWY12}, Leike and Heizmann propose
linear ranking templates~\cite{DBLP:conf/tacas/LeikeH14}, whereas Bradley et
al.~compute lexicographic linear ranking functions supported by inductive
linear invariants~\cite{DBLP:conf/cav/BradleyMS05}.

While the synthesis of termination arguments for linear programs over
integers is indeed well covered in the literature, there is very limited
work for programs over machine integers.  Cook et al.~present a method
based on a reduction to Presburger arithmetic, and a template-matching
approach for predefined classes of ranking functions based on reduction to
SAT- and QBF-solving~\cite{DBLP:conf/tacas/CookKRW10}.  Similarly, the only
work we are aware of that can compute nonlinear ranking functions for
imperative loops with polynomial guards and polynomial assignments
is~\cite{DBLP:conf/vmcai/BradleyMS05}.  However, this work extends only to
polynomials.

Given the lack of research on termination of nonlinear programs, as well as
programs over bit-vectors and floats, our work focused on covering these
areas.  One of the obvious conclusions that can be reached from
Figure~\ref{fig:handletable} is that most methods tend to specialise on a
certain aspect of termination proving that they can solve efficiently. 
Conversely to this view, we aim for generality, as we do not restrict the
form of the synthesised ranking functions, nor the form of the input
programs. 


As mentioned in Section~\ref{sec:intro}, approaches based on Ramsey's
theorem compute a set of local termination conditions that decrease as
execution proceeds through the loop and require expensive reachability
analyses~\cite{DBLP:conf/lpe/CodishG03,DBLP:conf/lics/PodelskiR04,DBLP:conf/pldi/CookPR06}.  In an attempt to reduce the complexity of
checking the validity of the termination argument, Cook et al.~present an
iterative termination proving procedure that searches for lexicographic
termination arguments~\cite{DBLP:conf/tacas/CookSZ13}, whereas Kroening et
al.~strengthen the termination argument such that it becomes a transitive
relation~\cite{DBLP:conf/cav/KroeningSTW10}. Following the same trend, 
we search for lexicographic nonlinear termination arguments that can be verified 
with a single call to a SAT solver.

Proving program termination implies the simultaneous search for a
termination argument and a supporting invariant.  Brock\-schmidt et
al.~share the same representation of the state of the termination proof
between the safety prover and the ranking function synthesis
tool~\cite{DBLP:conf/cav/BrockschmidtCF13}.  Bradley et al.~combine the
generation of ranking functions with the generation of invariants to form a
single constraint solving problem such that the necessary supporting
invariants for the ranking function are discovered on
demand~\cite{DBLP:conf/cav/BradleyMS05}.  In our setting, both the ranking
function and the supporting invariant are iteratively constructed in the
same refinement loop.

While program termination has been extensively studied, much less research
has been conducted in the area of proving non-termination.  Gupta et
al.~dynamically enumerate lasso-shaped candidate paths for counterexamples,
and then statically prove their
feasibility~\cite{DBLP:conf/popl/GuptaHMRX08}.  Chen et al.~prove
non-termination via reduction to safety
proving~\cite{DBLP:conf/tacas/ChenCFNO14}.  Their iterative algorithm uses
counterexamples to a fixed safety property to refine an under-approximation
of a program.  In order to prove both termination and non-termination,
Harris et al.~compose several program analyses (termination provers for
multi-path loops, non-termination provers for cycles, and global safety
provers)~\cite{DBLP:conf/sas/HarrisLNR10}.  We propose a uniform treatment
of termination and non-termination by formulating a generalised second-order
formula whose solution is a proof of one of them.


\bibliographystyle{splncs}
\bibliography{synth}{}

\end{document}